\documentclass[12pt]{article}

\usepackage[utf8]{inputenc}

\usepackage{color}
\usepackage{graphicx}
\usepackage{subfigure}
\usepackage[font=small]{caption} 
\usepackage{xspace}
\usepackage{amsmath,amssymb,latexsym}

\newcommand{\textem}[1]{{\em #1}}

\newtheorem{theorem}{Theorem}
\newtheorem{lemma}{Lemma}

\newtheorem{claim}{Claim}
\newtheorem{proposition}{Proposition}

\newenvironment{proof}{\noindent\textit{Proof.}}{{}\hfill $\Box$\\}

\newcommand{\card}[1]{\left|{#1}\right|}

\newcommand{\ceil}[1]{\left\lceil{#1}\right\rceil}
\newcommand{\paren}[1]{\left({#1}\right)}

\newcommand{\set}[1]{\left\{{#1}\right\}}

\let\d=\delta
\let\g=\gamma

\newcommand{\D}{\mathcal{D}}

\newcommand{\U}{{\cal U}}
\newcommand{\cU}{\mathcal{U}}
\newcommand{\cNE}{\mathcal{E}}
\newcommand{\cF}{\mathcal{F}}

\graphicspath{{./}}

\begin{document}

 \title{Self-Organizing Flows in Social Networks}

\author{
Nidhi Hegde\\
       {Bell Labs France Alcatel-Lucent}\\
       {France}\\
       \texttt{nidhi.hegde@alcatel-lucent.com}
\and
Laurent Massoulié\thanks{Part of this work was done while
at Technicolor.}\\
       {Microsoft Research -- Inria Joint Centre}\\
       {France}\\
       \texttt{laurent.massoulie@inria.fr}
\and
Laurent Viennot\thanks{Supported by the Inria project-team ``Gang'' in the ``LIAFA'' laboratory and by the ``LINCS'' laboratory.}\\
       {Inria -- Paris Diderot University}\\
       {France}\\
       \texttt{laurent.viennot@inria.fr}
}

\date{}

\maketitle

\begin{abstract}
 Social networks offer users new means of accessing information,
 essentially relying on ``social filtering'', i.e. propagation and
 filtering of information by social contacts. The sheer amount of data
 flowing in these networks, combined with the limited budget of
 attention of each user, makes it difficult to ensure that social
 filtering brings relevant content to interested users.  Our
 motivation in this paper is to measure to what extent
 self-organization of a social network results in efficient social
 filtering.
 
 To this end we introduce  {\em flow games}, a simple abstraction that
 models network formation under selfish dynamics, featuring
 user-specific interests and budget of attention. In the context of
 homogeneous user interests, we show that selfish dynamics converge to
 a stable network structure (namely a pure Nash equilibrium) with
 close-to-optimal information dissemination.
 We show that, in contrast, for the more realistic case of heterogeneous
 interests, selfish dynamics may lead to
 information dissemination that can be
 arbitrarily inefficient, as captured by an unbounded ``price of
 anarchy''.
 
 Nevertheless the situation differs when user
 interests exhibit a particular structure, captured by a metric space
 with low doubling dimension. In that case, natural autonomous dynamics
 converge to a stable configuration. Moreover, users obtain all the
 information of interest to them in the corresponding dissemination,
 provided their budget of attention is logarithmic in the size of their
 interest set.
 \end{abstract}

\section{Introduction}
\label{sec:intro}


Information access has been revolutionized by the advent of social networks such as Facebook, Google+ and Twitter. These platforms have brought about the new paradigm of ``social filtering'', whereby one accesses information by ``following'' social contacts.

This is especially true for twitter-like microblogging  social networks.
In such networks the functions of filtering, editing and disseminating news are totally distributed, in contrast to traditional news channels. 
%
The efficiency of social filtering is critically affected by the network topology, as captured by the contact-follower relationships. Today's networks provide recommendations to users for potentially useful contacts to follow, but don't interfere any further with topology formation. In this sense, these networks self-organize, under the selfish decisions of individual users. 

This begs the following question: when does such autonomous and selfish self-organizing topology lead to efficient information dissemination? The answer will in turn indicate under what circumstances self-organization is insufficient, and thus when additional mechanisms, such as incentive schemes, should be introduced.

Two parameters play a key role in this problem. 
On the one hand each user aims to maximize the coverage of the topics
of his interest.
On the other hand, a user pays with his attention: 
filtering interesting information from spam (i.e. information that does
not fall in his topics of interest) incurs a cost.
Users must therefore trade-off topic coverage against attention cost.
As pointed out by Simon~\cite{Simon:Organizations}, as information
becomes abundant, another resource becomes scarce: attention.

Furthermore, there is an interplay between participants in a social network where
filtering by one user may benefit another, inducing complex
dependencies in decisions on creating connections. 
To model this, we introduce a network formation game
called \emph{flow game} where 
some users produce news about specific topics and each user is interested in
receiving all news about a set of topics specific to him.
Each user is a selfish agent that can choose his incoming connections within
a certain budget of attention in order to maximize the coverage
of his set of topics of interest. 

This model is of interest on its own, as it enriches the class of existing
network formation games with a focus on flow dissemination under
bounded connections.
This model could also be of interest in the context of peer-to-peer
streaming and file sharing or publish/subscribe applications.

\subsection{Our results}

An important feature in our model is a user's budget of attention
for the consumption of content.  In previous
work~\cite{Jiang2013:BudgetofAttention} the budget of attention was
modelled as a limit on the rate with which a user consults a friend,
with a different objective of minimizing delay in receiving all content.
In the present work we are interested in a more fundamental question,
of how efficient social networks are formed in the first place.  We
consider the model where users are interested in specific subsets of topics
and their objective is to maximize the number of flows received corresponding
to these topics.   
As such, we  model the budget of attention as a constraint on the
number of connections a user may create (rather than a rate of consultation). 
Our aim is to build a simple model capturing the complexity of the problem.
This way of capturing the budget of attention amounts to assuming that
each connection consumes the same amount of attention. We discuss in 
Section~\ref{sec:costs} 
how we can tweak our model to more finely model attention consumption.

We capture users' interests in topics through user-specific values for
each topic and define the \emph{utility} a user receives to be the sum of values
of all received topics.  Each user's objective in a \emph{flow game} is
then to choose connections so as to maximize his utility.
We additionally assume that a user may produce news about one topic at most
even if he redistributes other topics. This is coherent with an
empirical study of twitter traces~\cite{DBLP:conf/icwsm/ChaHBG10} where it
is shown that ordinary users (as opposed to celebrities or newspapers) can 
gain influence by concentrating on a single topic.

Our main results relate to the stability and efficiency of the
formation of information flows.  We derive conditions where selfish dynamics 
converge to a pure Nash equilibrium.
We then give approximation ratios bounding the quality of
an equilibrium compared to an optimal solution.
This is traditionally measured through the price of anarchy,
the ratio of the global welfare (measured as the sum of
user utilities) at an optimal solution compared to that at 
the worst equilibrium.

More precisely, we first consider homogeneous games where all users are
interested in the same set of topics. We can then prove that selfish dynamics
always converge to an equilibrium. Selfish dynamics comprise of any sequence of
moves, where in each move a user is given the opportunity to selfishly rewire his
connections to increase his utility. We show moreover that convergence occurs
within a polynomial number of rounds where a
round is a sequence of selfish moves including at least one move per user.  
We additionally show that the price of anarchy is bounded and
approaches 1 as the budget of attention of users increases.

In the more general case where users interests are heterogeneous, selfish
dynamics may not converge and price of anarchy may be unbounded.  
However, we observe
that fast convergence towards efficient configurations can occur when users'
interests are captured by a metric space with sufficient structure.  Here, the
interests of a user are modeled a point in this space such that nearby
topics are of interest to the user.  Sufficient structure typically arises
when the metric space is a Euclidean space with low dimension. Our results are
tailored to the more general case of metrics with low doubling dimension.  Low
dimension assumptions are classically used in information retrieval when data
can be viewed as a matrix which is approximated with a low rank matrix. For
example, a ranking technique for the web is proposed in~\cite{haveliwala2003}
using a 16-dimensional space for representing topics of web pages. Closer to the
context of our study, modeling people's opinions as points in a low Euclidean
space is a classical approach in social sciences. Political spectrum for example
is often modeled as a one dimensional space along a Left-Right axis. 
consist in introducing more dimensions. 
This concept can be formalized with single-peaked
preference curves~\cite{duncan}. An online system for exchanging political views
could be a concrete example where the technical conditions of our model are met.
We believe that the same applies for the various domains of interest of a user,
implying that our model remains valid more broadly if we can attribute several
points of interest (one per domain) to each user.  An extension of our model in
that direction is proposed.

\subsection{Related work}
\label{sec:rel}

Information spread in networks has been studied extensively. 
Much of the past work study the properties of information
diffusion on given networks with given sharing protocols.  Our goal in
this work is to study how networks form when users create connections
with the objective of efficient content dissemination 
in a game-theoretical approach.
This work thus follows the large amount of work in network 
formation games. However, to the best of our knowledge, the
objective of efficient information dissemination under edge
constraints and interest sets that we
consider here is novel.  We now discuss some work in those domains
that are most relevant to this paper.

A lot of attention has been given to simple models of diffusion
in social networks
such as ``rumor spreading'' or ``cascading'' 
where a piece of information interests all
users and is propagated in the network through 
random interactions typically (see e.g. the related work mentioned 
in~\cite{Giakkoupis2014} for rumor spreading and~\cite{kempe2003}
for cascading). 
In this paper, we are more interested in the selective propagation
of information according to connections chosen locally for optimizing
the coverage of personal interests.

Network formation games have been considered in previous work in
economics and in the context of the formation of Internet peering
relations and peer-to-peer overlay networks.  Economic models of
network formation~\cite{Jackson2010_SocialEconomicNets} use edges to
represent social relations and it is typically assumed that the
creation of an edge needs bilateral agreement since both users benefit
from an edge.  
Our model is oriented and unilateral agreement is more relevant to the
notion of \emph{following} in social networks.  A non-cooperative
one-way link connection game has been considered in previous
work~\cite{Bala2000:NoncoopNetFormation}, where each created link
incurs a cost and users are interested in connecting to all other
users.  Our model is richer and more realistic where we consider connections to
subsets of information flows that hold user-specific intrinsic values.

Network creation games in the context of the Internet have been
considered~\cite{Nisan07_AGT}, where distributed formation of
undirected edges with a linear cost on each edge formed is studied.
In such games, each user's objective is to minimize total formation
cost while either minimizing distance to all other users
\cite{Fabrikant03_NetworkCreation}, or ensuring connection to a given
subset of nodes~\cite{Anshelevich:2004:PoSNetworkDesign}.   We
consider a bound on edge costs, in the form of a limit on the number
of in-edges at each node, and further, we focus on connections that allow
specific flows of information.

Interestingly, bounded budget network formation games have already been
considered.  Bounded budget connection
games~\cite{Laoutaris:2008:BBCGames} consider a bound on each user's
budget in creating edges, with the objective being the minimization of
the sum of weighted distances to other nodes.  A similar model is
considered in~\cite{Bei:2011:BoundedBetweenness} where each user's
objective is to maximize his influence, measured using betweenness
centrality.  In our work however, rather than minimizing distance to
any node, we consider a formation game with the objective of ensuring
connections to a subset of flows of interest, without regard to the
particular nodes.


The notion of connecting to users that can provide a content flow of
interest is similar to peer-to-peer live streaming
systems~\cite{Massoulie08_p2p}.  Unlike peer-to-peer streaming, 
we do not aim to satisfy flow rates, rather our aim is to connect
to as many sets of relevant flows as possible.  Moreover, our model
allows differing user interests.
%
%
The stability of connecting users of a peer-to-peer network according
to some affinities between users was studied 
using b-matching and acyclic preference
systems~\cite{Gai07_acyclic}. As a generalization of the stable
marriages problem, those systems consider configurations of undirected
edges based on mutual acceptance of an edge, whereas unilateral decision
is more suitable in our model. 
Our model is more intricate in the sense
that connections are based not only on preferences 
but also on complementarity of content obtained through various connections.

Most notably, a model similar to ours has been independently developed 
in~\cite{augustin2014}. The authors propose a model for explaining how
social media can provide efficient filtering of information. They model
online exchange of media information with three type of actors: official
media sites that provide fresh news every day, bloggers that relay some of these
news and users that access these news through bloggers. The model also includes
a game-theoretical part where the players are the bloggers. 
The strategy of a blogger 
is the set of news he decides to relay and his utility is the number of users
following him. Conversely to our model, this is not a network formation
game. However, it could be interesting to see if both approaches can be mixed
together for modeling multi-hop relaying where bloggers can also 
relay other bloggers.

In Section~\ref{sec:doubling} we model the space of user interests by
a metric space with low doubling dimension. Modeling interests of users
through a metric space seems a natural approach and bounded growth
metrics, or more generally doubling metrics, have shown to 
be very a general model~\cite{DBLP:journals/mst/PlaxtonRR99}
that can capture general situations, while still providing an algorithmic 
perspective.
The doubling dimension extends the
notion of dimension from Euclidean spaces to arbitrary metric
spaces. It has proven to be useful in many application domains such as 
nearest neighbor queries to databases~\cite{clarkson1999nearest},
network construction~\cite{DBLP:conf/soda/AbrahamMD04},
closest server selction~\cite{DBLP:conf/stoc/KargerR02}, etc.
Doubling metrics have notably been used to model distances in
networks such as Internet~\cite{flv08}.

\subsection{Organization of the paper}
Section~\ref{sec:model} introduces the model.  We study the case of
homogeneous interests in Section~\ref{sec:hom}.  
The heterogeneous case in its full
generality is considered in Section~\ref{sec:het} which details some
negative results.  Section~\ref{sec:doubling} is dedicated to the
specific scenario where users' interests are captured by a doubling
metric, enabling some positive results.
Section~\ref{sec:costs} presents how the costs of attention can be better
modeled with respect to the intersection of user interests.
We finally conclude in
Section~\ref{sec:conc} describing potential extensions of the current
work.

\section{Model}
\label{sec:model}

We consider a social network where users interested in some set of
content topics (or subjects) connect to (or \emph{follow} in social
networking parlance) other users in order to obtain such contents,
materialized by flows of news.  Each user may produce news for at most one
topic (but may forward news from other topics she is interested in). 
To distinguish the role of publisher from that of follower, we technically
assume that news concerning a given topic (or subject) are produced at
a given node called producer which is identified with that topic.

A \emph{flow} game is defined as a tuple $(V,P,S,\Delta)$ where
$V$ is a set of users, $P$ a set of producers (or subjects or topics)
and $S:V\rightarrow P$ is a function associating to each user $u$ 
its interest set $S_u\subseteq P$, and $\Delta:V\rightarrow \mathbb{N}$
is a function
associating to each user $u$ its budget of attention $\Delta_u$.
We let $n=\card{V}$ and $p=\card{P}$ denote the number of users and
producers respectively.
A flow game is \emph{homogeneous} if all users have the same interest
set: $S_u=P$ for all $u\in V$. If this is not the case, the game
is said to be \emph{heterogeneous}.

A strategy for user $u$ is a subset $F_u$ of $\set{(v,u) : v\in V\cup P}$
such that $\card{F_u}\le \Delta_u$
($\Delta_u$ is an upper bound on the in-degree of $u$ that we call the
\emph{budget} of $u$).
For all $(v,u)\in F_u$, we say that $u$ \emph{follows} $v$ or equivalently
that $u$ is connected to $v$ (such a link $(v,u)$ created by $u$ 
is oriented according to the
data flow, that is from $v$ to $u$). The collection $F=\set{F_u : u\in V}$ forms a
network defined by the directed graph $G(F)=(V\cup P,E(F))$ where
$E(F)=\cup_{u\in V}F_u$. A user $u$ is \emph{interested} in a subject $s$
if $s\in S_u$. A user $u$ \emph{receives} a subject $s\in P$
if there exists a directed path from $s$ to $u$ in $G(F)$
such that all intermediate nodes are interested in $s$.
We allow for a natural \emph{filtering} mechanism, where a user retransmits only subjects 
she is interested in. For a given configuration, we let $R(u)$ 
denote the set of subjects received by $u$.
The utility $U_u(F)$ for user $u$ is the number of subjects in $S_u$ she
receives, that is $U_u(F)=|R(u)\cap S_u|$. 
The utility of $u$ is maximized if $U_u(F)=\card{S_u}$.

We denote by \emph{move}, a shift from a collection $F$ of strategies
to a collection $F'$ where a single user $u$ changes her strategy from
a set $F_u$ to another $F_u'$. (We say that $u$ rewires her connections.)
The move is \emph{selfish} if $U_u(F')>U_u(F)$.
\emph{Selfish dynamics} (or dynamics for short) are the sequences
of selfish moves. We say that dynamics \emph{converge} if any sequence
of selfish moves is necessarily finite.
The network is at equilibrium (or stable) if no selfish move is possible.
In standard game-theoretic terminology, this corresponds to a pure
Nash equilibrium. 
The \emph{global welfare} of the system
is defined as the overall system utility: $\cU=\sum_{u \in V} U_u$.
The  efficiency of selfish, self-organization of a game is classically 
captured by the notion of price of anarchy defined as the ratio
of the optimal global welfare over the global welfare
of the worst equilibrium: $\text{PoA}= \frac{\max_{F\in \cF}\sum_{u\in V}U_u(F)}{\min_{F \in \cNE}\sum_{u\in V}U_u(F)}$,
where $\cF$ denotes the set of possible collection of strategies and
$\cNE\subseteq \cF$ denotes the set of equilibria.

In some of our proofs we make use of the notion of potential functions.  An ordinal (or general~\cite{DBLP:conf/stoc/FabrikantPT04}) potential
function~\cite{DBLP:bibsonomy_Monderer:1996p121} is a function
$f:\cF\rightarrow\mathbb{R}$ such that
$\mathop{sign}(f(F')-f(F))=\mathop{sign}(U_u(F')-U_u(F))$ for any move
from $F$ to $F'$ where user $u$ changes her strategy.  If
$f(F')-f(F)=U_u(F')-U_u(F)$, $f$ is called an exact potential
function.  This notion was introduced by Monderer and
Shapley~\cite{DBLP:bibsonomy_Monderer:1996p121} who show that it is
tightly related to the notion of a congestion game~\cite{rosenthal}.
The use of potential functions is a standard technique 
to show convergence of dynamics
and to bound price of
anarchy~\cite{DBLP:conf/stoc/FabrikantPT04,Roughgarden:2006p111}.

\section{Homogeneous interests}
\label{sec:hom}

We first consider the case where 
identical sets of interests, $S_u=P$, for all $u\in V(G)$.  In this context, we
first analyze how to achieve optimal global welfare before establishing an upper
bound on the price of anarchy. We will then show convergence of dynamics
and provide a polynomial bound on convergence time.



\subsection{Optimal Utility and Optimal Global Welfare}
\label{sub:welfare}

We first analyze what is the optimal utility, i.e. the maximum utility a user
can get, and compare it to the optimal global welfare, i.e. the sum of user
utilities obtained
under an optimal centrally designed configuration.  

First consider the maximum utility a given user $u$ can get. Clearly,
$u$ cannot achieve utility larger than $p$, which corresponds to
obtaining all the subjects in $P$. Moreover, he cannot obtain more
subjects than the aggregate budget of attention of all users, that is
$\sum_{u\in V(G)}\Delta_u = n\overline{\Delta}$, where $\overline{\Delta}$ 
is the average in-degree per node.
More precisely, when a user $v$ receives subjects obtained by a user
$w$, some link must connect $w$ to $v$ or to a user $v$ is connected to by 
some path.
Overall, at least one link per user $w\not= u$ must be consumed 
with a user-to-user connection and cannot be used to retrieve a subject
from a producer. 
In a configuration where $u$ receives 
a maximum number of subjects, his utility is thus
at most $(\sum_{u\in V(G)}\Delta_u)-(n-1) =1+\sum_{u\in V(G)}(\Delta_u-1)$. 
Note that this bound is achieved in a singly linked chain configuration
where users are placed along a chain ($u$ being the last node) 
and where each user follows the previous user in the
chain and use remaining connections to follow producers. 
We thus get the following claim.

\begin{claim}
\label{claim:single}
In an homogeneous flow game with $p$ producers and
$n$ users with average budget $\overline\Delta$, 
the optimal utility
a given user can get among all configurations is 
$U^* = \min\left(p, 1 + n(\overline{\Delta}-1)\right)$.
\end{claim}

On the other hand, all users can receive the same set of 
$\min\left(p, n(\overline{\Delta}-1)\right) \ge U^*-1$ 
subjects in a ring configuration
which consists in
forming an oriented ring between users and connecting all remaining
connections to pairwise distinct producers as depicted
by Figure~\ref{fig:homring}. This shows that the average utility a user
gets at optimal global welfare equals the optimal utility $U^*$ up to one.

\begin{figure}[htbp]
\centering
\includegraphics[width=.6\linewidth]{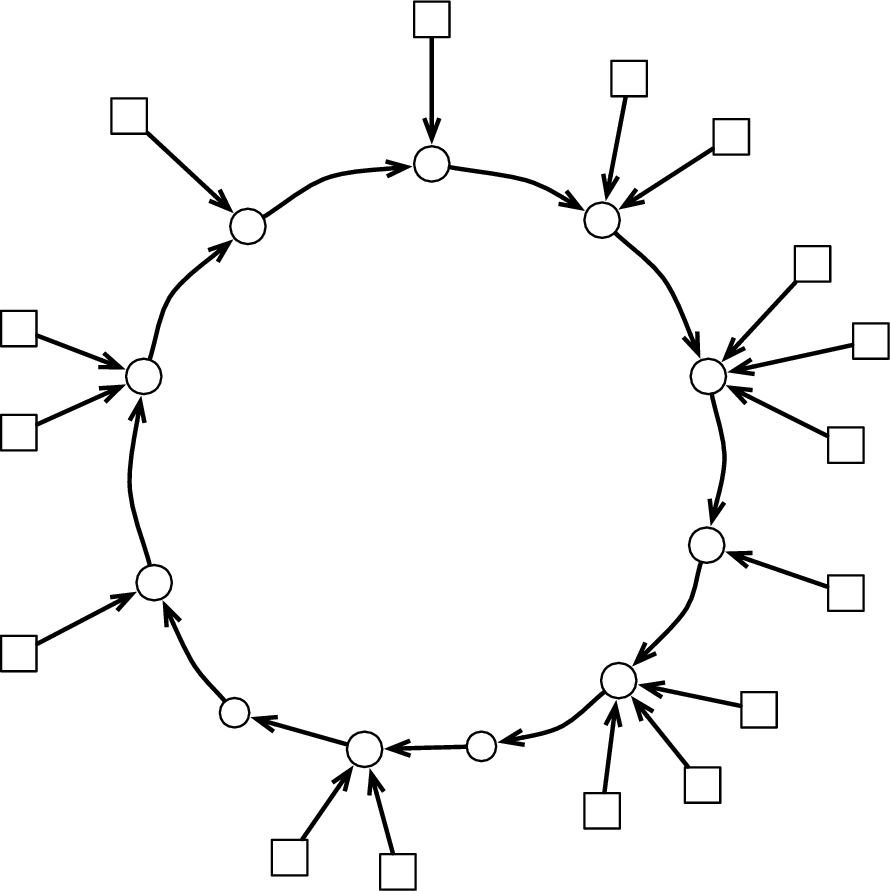}
\caption{An optimal configuration in the homogeneous case (circles and squares
  represent users and producers respectively and edges point in the direction of
  content transfer, as in all figures throughout the paper). It is called a
  \emph{ring} configuration as users form an oriented
  ring.}
\label{fig:homring}
\end{figure}

We will prove that optimal global welfare is generally obtained with 
a ring configuration. However, the singly linked chain can be optimal
in some special cases. For example,
in a flow game with $n=2$ users having budget 2, a singly linked
chain achieves a global welfare of 5 compared to 4 for the ring. With three
users, both configurations achieve a global welfare of 9. 
More generally, we prove the following.

\begin{claim}
\label{claim:opt}
Consider an homogeneous flow game with $n\ge 3$ users
where each user has budget 2 or more 
and where the average budget $\overline\Delta$ is less than $\frac{p}{n}+1$
(equivalently $p\ge 1 + n(\overline{\Delta}-1)$).
The optimal global welfare is that of a ring configuration that is
$n^2(\overline{\Delta}-1)=n(U^*-1)$.
\end{claim}

\begin{proof}
The condition on average budget can be written
$p > n(\overline{\Delta}-1)$. We first assume $p>n\overline{\Delta}$
as this implies that there always exists some producer
not received by a user as $n\overline{\Delta}$ is a clear upper-bound of the
number of subjects that can be globally retrieved by the users.
Consider a configuration providing maximum global welfare.

The graph of user-to-user connections must be connected. If there were two
different connected components $a$ and $b$, we could find a user $v$ in $b$
following some producer $x$. We could then rewire the producer-to-user
connections in $a$ so that some node $u$ in $a$ receives $x$ plus some subject
$y$ not received by $v$ while preserving the utility of each user (as nodes in
$a$ have budget at least 2, they receive at least two subjects).  Rewiring
$(x,v)$ into $(u,v)$ would then increase the global welfare of 1 at least.

Now consider the strongly connected components of this graph resulting
from user-to-user connections: we call component a maximal set $W$
of users such that there is a path from $s$ to $t$ for all $s,t\in W$.
Each component must be a singleton or a ring.
This is because any strongly connected graph with a minimum number of links
must be a ring (one link per user is necessary if the graph has two nodes or
more, and it is sufficient only when these links form an oriented ring).
If a non-singleton component was not a ring, we could 
rewire its internal connections to form a ring
and use the saved connections to follow more producers and increase
the global welfare. 

We first show that if one component is a ring, then it is the only component.
We show this by proving that there cannot exist
a connection between a ring component and any other component.
If a ring $b$ is connected to a ring $a$, then the two rings
can be merged, allowing nodes in $a$ to receive subjects received in $b$
and not in $a$ previously.
If a singleton component $\set{v}$ is connected to a 
ring $a$, then integrating $v$ in the ring allows all users in $a$ to
receive the other subjects received by $v$.
Finally, suppose a ring component $b$ 
is connected to a singleton component $\set{u}$
through a connection $(u,v)$ where $v$ is a node in $b$.
As $u$ has budget at least 2 and connects to at most one component,
it has at least one connection to a producer $x$.
This connection can be used to integrate $u$ in the ring of $b$ while saving
the previous ring connection of $v$. (If the node preceding $v$ in the ring
is $w$, we replace $(x,u)$ by $(w,u)$ and release connection $(w,v)$.) 
We can thus connect $v$ to $x$ and
get a configuration with higher welfare since $u$ now receives more subjects.

We now consider the special case where all components are singletons.  If a user
$a$ was followed by $b$ and $c$, we could increase the global welfare by letting
$c$ follow $b$ instead of $a$. In such an optimal configuration, each user is
followed by at most one user. The graph of user-to-user connections is thus a
tree. Highest utility is obtained when the users follow pairwise disjoint sets
of producers.  If a node $a$ follows $b$ and $c$, it then more profitable to let
a leaf of the sub-tree rooted at $b$ follow $c$ instead of some producer $x$ and
let $a$ follow $x$ instead of $c$. The tree must thus be a singly linked
chain. Such a configuration cannot achieve global welfare higher than that of a
ring configuration for $n\ge 3$ and nodes of budget at least 2.

Finally consider the case where $p > n(\overline{\Delta}-1)$ and 
$p\le n\overline{\Delta}$.  Consider the modified 
flow game with the same set of users (with the same budgets) and with
$p'=p+n>n\overline\Delta$ producers ($n$ more producers are added).
Any optimal configuration for the modified game can be transformed
into a configuration for the original game since 
$n(\overline{\Delta}-1) < p$ producers at most get connected.
The optimal configuration we consider in the original flow game must thus 
achieve a global welfare as high as what can achieved in the modified game.
It is thus optimal in the modified game also. We can thus again conclude
that the optimal global welfare is that of a ring, that is 
$n^2(\overline{\Delta}-1)=n(U^*-1)$.
\end{proof}

Interestingly, we have indeed shown that for $n\ge 4$, the ring is the only
optimal configuration when users have budget at least 2.
When some users have budget 1 and yet 4 or more nodes have budget at least 2, 
any optimal configuration still has a ring structure. 
As a connection to a budget 1 user is
equivalent to a connection to the node he follows, we can virtually
erase budget 1 users as transparent forwarding nodes along connections
between users with budget at least 2 and producers. The virtual configuration
must be optimal also and must be a ring. However,
there are several ways budget 1 users can be connected to other users. 
They can be
connected to any node of the ring, or form trees connected to some nodes
of the ring, or form chains that replace some links of the ring. (A singly
linked ring is one possible optimal configuration.)

\subsection{Price of Anarchy}
\label{sub:poa}

We now consider a distributed setting where each user selfishly
rewires his incoming connections if he can improve his utility,
i.e., if this allows him to receive more subjects.
The following proposition shows that with homogeneous user interests
and budget of attention at least 3, 
self organization is efficient if dynamics converge, 
achieving a price of anarchy close to $1$.

\begin{proposition}
\label{prop:nash_utility}
Assume that $3\le \Delta_u< p$ for every user $u \in V$ of
a homogeneous flow game with $n\ge 3$ users. Then under any
equilibrium the utility of a user is at least 
$\frac{\overline{\Delta}-2}{\overline{\Delta}-1} (U^*-1)$ 
where $U^*$ is his optimal
utility. 
The price of anarchy
is $1+1/(\overline{\Delta}-2)$ at most, approaching 1 for large
$\overline{\Delta}$.
\end{proposition}

Before proving Proposition~\ref{prop:nash_utility}, we establish two lemmas.  
The first one allows to show the existence of strongly connected components at equilibrium showing that under some technical assumption, if a forwarding path
exists at equilibrium, then a reverse path should also exists. 
. 

\begin{lemma}
\label{lem:pathback}
If an equilibrium is reached such that there exists a path $x,u_1,\ldots,u_k$ 
where $x$ is a producer, $u_k$ has in-degree bound $\Delta_{u_k}\ge 3$ 
and a producer $y$ is not received by $u_k$, 
then there is a path from $u_k$ to $u_1$. 
\end{lemma}

\begin{proof}
The existence of the path $x,u_1,\ldots,u_k$ first implies that $R(u_1)\subset R(u_k)$.
 Since $\Delta_{u_k}\ge 3$, $u_k$ must be connected to two nodes $v$
 and $w$ distinct from $u_{k-1}$.  We first claim that $v$ must bring
 at least one unique subject $z_1$ (not in $R(u_1)$ and thus different from $x$), 
otherwise, $u_k$ could unfollow $v$ and
 follow $y$ instead. Similarly, $w$ must bring at least one unique subject $z_2$ (different from $z_1$ and not in $R(u_1)$). 
 Then if there is no path from $u_k$ to $u_1$, $u_1$ would unfollow $x$ and follow
 $u_k$ instead, so that he only loses one subject $x$ but gains at least two
 subjects $z_1$ and $z_2$.
\end{proof}

The second Lemma will be used to bound the number of links between users
in a strongly connected component at equilibrium. We call \emph{transitivity
arc} a link $(s,t)$ such that there exists a path from $s$ to $t$.
Such a link is useless as any subjects it brings is also provided by the
path, and node $t$ would be better off following a non-received producer instead 
of $s$. Such links cannot thus exist at equilibrium.

\begin{lemma}
\label{lem:strong}
Consider a strongly connected graph $G$ with $n$ nodes and $m$ arcs
(multiple arcs are allowed). If
$m \ge 2n-1$, then $G$ contains a transitivity arc.
\end{lemma}

\begin{proof}
We prove the result by induction on $n$. The hypothesis is true for $n=1$ (a
self-loop on vertex $s$ is a transitivity arc for the empty directed path from
$s$ to $s$).  We denote by $n(G)$ the number of nodes in the graph $G$ and by
$m(G)$ the number of edges in the graph $G$.  Now consider $n>1$ and assume that
the property is true for any
graph $G'$ with $n(G')<n$. Consider a strongly connected 
graph $G$ with $n$ nodes containing no transitivity arc.  
Since $n\ge 2$, $G$ must contain a circuit, i.e. an oriented cycle,
with $k\ge 2$ nodes. The only arcs connecting two nodes of the
circuit are the circuit arcs (otherwise, we would encounter a
transitivity arc).
Consider the graph $G'$
obtained by contracting the circuit to one node.
We have $m(G')=m(G)-k$ and $n(G')=n(G)-k+1<n$. Note that  $G'$ does not contain
a transitivity arc either. Our induction
hypothesis thus implies that $m(G')<2n(G')-1$. That is
$m(G)-k<2(n-k+1)-1$ or equivalently $m(G)<2n-k+1\le 2n-1$ as
$k\ge 2$.  The property is thus satisfied for $n$.
\end{proof}

We are now ready to prove Proposition~\ref{prop:nash_utility}.

\begin{proof}[of Proposition~\ref{prop:nash_utility}]
 Consider any equilibrium.
If all users receive at least $p$ subjects, then the equilibrium is optimal.  We thus consider the case where there is a user $u$ who receives less than $p$ subjects.
Then $u$ must be connected to some producer $x$ by a path
$x,u_1,\ldots,u_k=u$. 
Consider
the graph $G'$ induced by users reachable from $u_1$ that receive less
than $p$ subjects.
By Lemma~\ref{lem:pathback}, $G'$ is strongly connected and all its users
receive the same number $p'<p$ of subjects.

We claim that two users $u$ and $v$ of $G'$ cannot follow the same
producer $y$. As there exists a path from $u$ to $v$, the link 
$(y,v)$ would be redundant and $v$ would be better off following some unreceived subject instead.
Moreover, the fact that users in $G'$ do not receive all subjects implies
that they have spent all their budget of attention.
We thus conclude that the number of edges in $G'$
is $m(G')=\sum_{u\in V(G')}\Delta_u - p'$.
As the network is stable, there is no transitivity arc in $G'$.  
(Otherwise, a transitivity arc $(s,t)$ would be redundant with some 
path from $s$ to $t$, and $t$ would be unstable as he could increase his utility by rewiring
this link to a new producer.)
Lemma~\ref{lem:strong} thus implies $m(G')\le 2n(G')-2\le 2n(G')$,
where $n(G')$ is the number of nodes in $G'$.
We thus get 
$p' \ge \sum_{u\in V(G')}\Delta_u-2n(G') = \sum_{u\in V(G')} (\Delta_u-2)$.

First consider the case $p'\le p-2$. Suppose there exists
a user $w\notin V(G')$. 
As $\Delta_w\ge 3$, 
$w$ has utility at least $p'+2$ since he can gather the $p'$ subjects 
received in $G'$ plus two others by connecting to one node in $G'$ 
plus the two corresponding producers. He thus receives two subjects not received
in $G'$ but this contradicts the stability as $u_1$ would better 
unfollow $x$ and connect to $w$.
We thus conclude that $G'$ indeed contains all users,
implying $p'\ge n(\overline{\Delta}-2)$.
Using Claim~\ref{claim:single}, the utility of each user is at least
$p'\ge \frac{\overline{\Delta}-2}{\overline{\Delta}-1} (U^*-1)$.
The global welfare at equilibrium is thus at least 
$n^2(\overline\Delta-2)$. As the optimal global welfare is at most
$n^2(\overline{\Delta}-1)$ according to Claim~\ref{claim:opt},
the price of anarchy is at most 
$\frac{\overline{\Delta}-1}{\overline{\Delta}-2}$.

Finally, in the remaining case where $p'= p-1$, some users may
be outside $V(G')$. However such users must also receive
$p-1$ subjects at equilibrium (if a user was receiving less, 
he could increase his utility by following a node in $V(G')$).
The utility of each user is thus at least
$\frac{p-1}{p} U^*
  \ge \frac{\overline{\Delta}-2}{\overline{\Delta}-1} (U^*-1)$
as $p\ge \overline{\Delta}-1$
and $U^* \ge (U^*-1)$.
As the optimal global welfare is bounded by $nU^*$,
the price of anarchy is at most 
$\frac{p}{p-1} \le \frac{\overline{\Delta}-1}{\overline{\Delta}-2}$.

In both cases, each user gets utility 
$\frac{\overline{\Delta}-2}{\overline{\Delta}-1} (U^*-1)$ at least
at equilibrium and the price of anarchy is at most
$\frac{\overline{\Delta}-1}{\overline{\Delta}-2} 
  = 1 + \frac{1}{\overline{\Delta}-2}$.
\end{proof}

Note that the above proposition is tight in the sense that high
price of anarchy can arise when most of the users have budget only 2.
Figure~\ref{fig:poa.d2} presents the extreme configuration where
all nodes have budget 2. 
In this particular example, a doubly linked
chain forms a 
pure Nash equilibrium gathering only two subjects in total while a
ring configuration gathers $n$ subjects. The price of anarchy is thus $n/2$.
Indeed the doubly linked chain is still stable when some nodes have budget
more than 2 and use their spare connections to gather fresh subjects.
In that case, the price of anarchy is 
$\frac{n(\overline{\Delta}-1)}{n(\overline{\Delta}-2)+2}$. It thus remains
unbounded as long as $\overline{\Delta}=2+o(1)$.

\begin{figure}[htbp]
\centering
  \subfigure[Optimal configuration]{
  \label{fig:poa.d2.opt}
  \includegraphics[width=.6\linewidth]{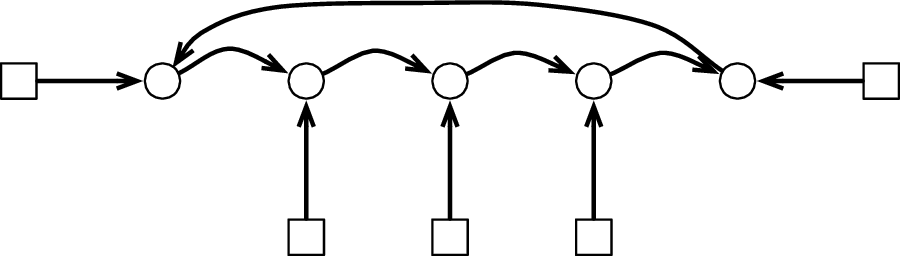}
  }\\%
  \subfigure[A pure Nash equilibrium configuration]{
  \label{fig:poa.d2.ne}
  \includegraphics[width=.6\linewidth]{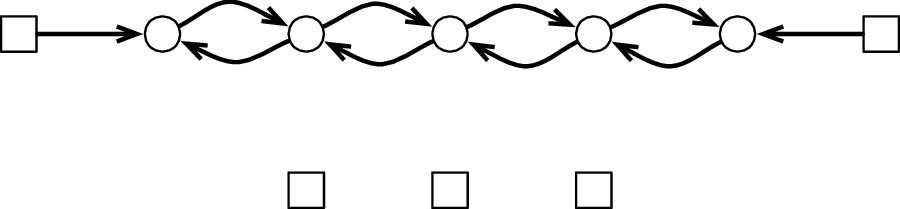}
  }
\caption{Two stable configurations in a
homogeneous game where all users have degree $\Delta=2$ showing
that price of anarchy can be $n/2$.}
\label{fig:poa.d2}
\end{figure}

\subsection{Convergence of Dynamics}
\label{sub:convhom}

We have thus shown that stable configurations of self-organizing
networks with homogeneous user interests are efficient.  However, do
network dynamics converge to an equilibrium ?
The following proposition answers this question in the
affirmative.  


\begin{proposition}
\label{prop:hom.conv}
Any homogeneous flow game has an ordinal potential function, implying that
selfish dynamics always converge to an equilibrium in finite time.
\end{proposition}

\begin{proof}
 Let $n_i$ denote the number of users that receive $i$ subjects and
 consider the sequence $(n_{0}, n_{1}, \ldots, n_{p})$.
 We show that this sequence
 always decreases according to lexicographic ordering when users make 
 selfish moves. 
 The function $-\sum_{0\le i\le p}n_i\,n^{p-i}$ 
(obtained by reading $n_0 n_1\cdots n_p$ as a number) 
is thus a potential function 
 that will always increase until
 a local maximum is reached, proving convergence to an equilibrium.

 Consider a user $u$ that is receiving $i$ subjects and that will
 make a selfish move to receive $j> i$ subjects
 instead.  
 Note that there is no path from $u$ to any other user receiving $k <
 i$ subjects.  Therefore any change by $u$ will not affect these
 users.  Now consider any user $v$ with $k \ge i$ subjects.  If there
 is no path from $u$ to $v$ then $u$'s selfish move does not affect
 $v$.  If there is such a path, then $v$ will now receive at least
 $j>i$ subjects. We thus now have $n_i-1$ users receiving $i$
 subjects, and the sequence $(n_{0}, n_{1}, \ldots, n_{p})$ has
 decreased according to lexicographic ordering.
\end{proof}

Combining Proposition~\ref{prop:nash_utility} and
Proposition~\ref{prop:hom.conv}, we obtain:

\begin{theorem}
\label{th:hom}
In a homogeneous flow game where $n\ge 3$ users have 
budget of attention at least 3,
less than $p$, and $\overline{\Delta}$ in average, 
selfish dynamics converge to an equilibrium such
that the utility of a user is at least 
$\frac{\overline{\Delta}-2}{\overline{\Delta}-1} (U^*-1)$ 
where $U^*$ is the optimal utility he can get. The price
of anarchy is $1+1/(\overline{\Delta}-2)$ at most.
\end{theorem}

Our proof of Proposition~\ref{prop:hom.conv} yields a very loose bound of $n^{p+1}$ on convergence time.  We
leave as an open question whether exponential time of convergence can
really arise. 
However, in the following proposition 
we show that a homogeneous flow game with at
least 4 subjects, a user with budget of attention at least 2 and a
user with budget of attention at least 3, is not equivalent to a
congestion game.  
%
This rules out the possibility of using techniques
similar to \cite{DBLP:conf/stoc/FabrikantPT04} to find equilibria in
polynomial time, and more generally to easily bound convergence time.  

\begin{proposition}
\label{prop:not.congestion}
Any homogeneous flow game  with at
least 4 subjects, a user with budget of attention at least 2 and a
user with budget of attention at least 3,
does not admit an exact potential function.
\end{proposition}

Note that a game is equivalent to a congestion game if and only if it admits an
exact potential function~\cite{DBLP:bibsonomy_Monderer:1996p121}.

\begin{proof}
This is proven by considering cycles in the strategy space where each
point corresponds to a set of strategies chosen by all the users
and an arc corresponds to a selfish move by a user.
A potential function assigns a value to each point. Its variation 
along an arc is the difference between the values assigned to the destination
and the source. We define its variation along a path as the
sum of the variations of the arcs of the path. Obviously, the variation
along a cycle must be zero.

Additionally, an exact potential function should ensure that the variation
during a move by a user $u$ equals the variation of the utility of $u$.
We define similarly the variation of utility along a path with
selfish moves from users $u_1,\ldots,u_k$ respectively as the sum of
utility variations for node $u_1$ in the first move, node $u_2$ in the
second move, and so on. If ever, we can exhibit a cycle with non-zero
utility variation in our flow game, it is clearly impossible to design an
exact potential function for that game satisfying both requirements.
(For more details about congestion games and exact potential functions,
see for example~\cite{DBLP:bibsonomy_Monderer:1996p121}.)

To show the proposition, we exhibit a 4-cycle in the strategy space
with non-zero utility variation along the cycle. Without loss of
generality, the game contains four producers $\set{a,b,c,d}$ and two
users $u,v$ with $\Delta_u\ge 2$ and $\Delta_v\ge 3$ as depicted
in Figure~\ref{fig:cycle4}. User $u$ can adopt in particular
strategy $A=\set{(a,u)}$ or $B=\set{(b,u),(c,u)}$. User $v$ can adopt
in particular strategy $C=\set{(u,v),(b,v),(c,v)}$ or $D=\set{(u,v),(d,v)}$.
Consider the cycle $(A,C)\rightarrow (B,C)\rightarrow (B,D)\rightarrow
(A,D)\rightarrow (A,C)$ where user $u$ moves from strategy $A$ to $B$
increasing his utility by 1, then $v$ moves from $C$ to $D$ and increases
his utility by 1, then $u$ moves back to $A$ with a utility variation of -1, and
finally $v$ moves back to $C$ increasing its utility by 1 again
(the strategies for other users remain fixed).
The overall sum is thus $2\not= 0$.
\end{proof}

\begin{figure}[htbp]
\centering
\includegraphics[width=\linewidth]{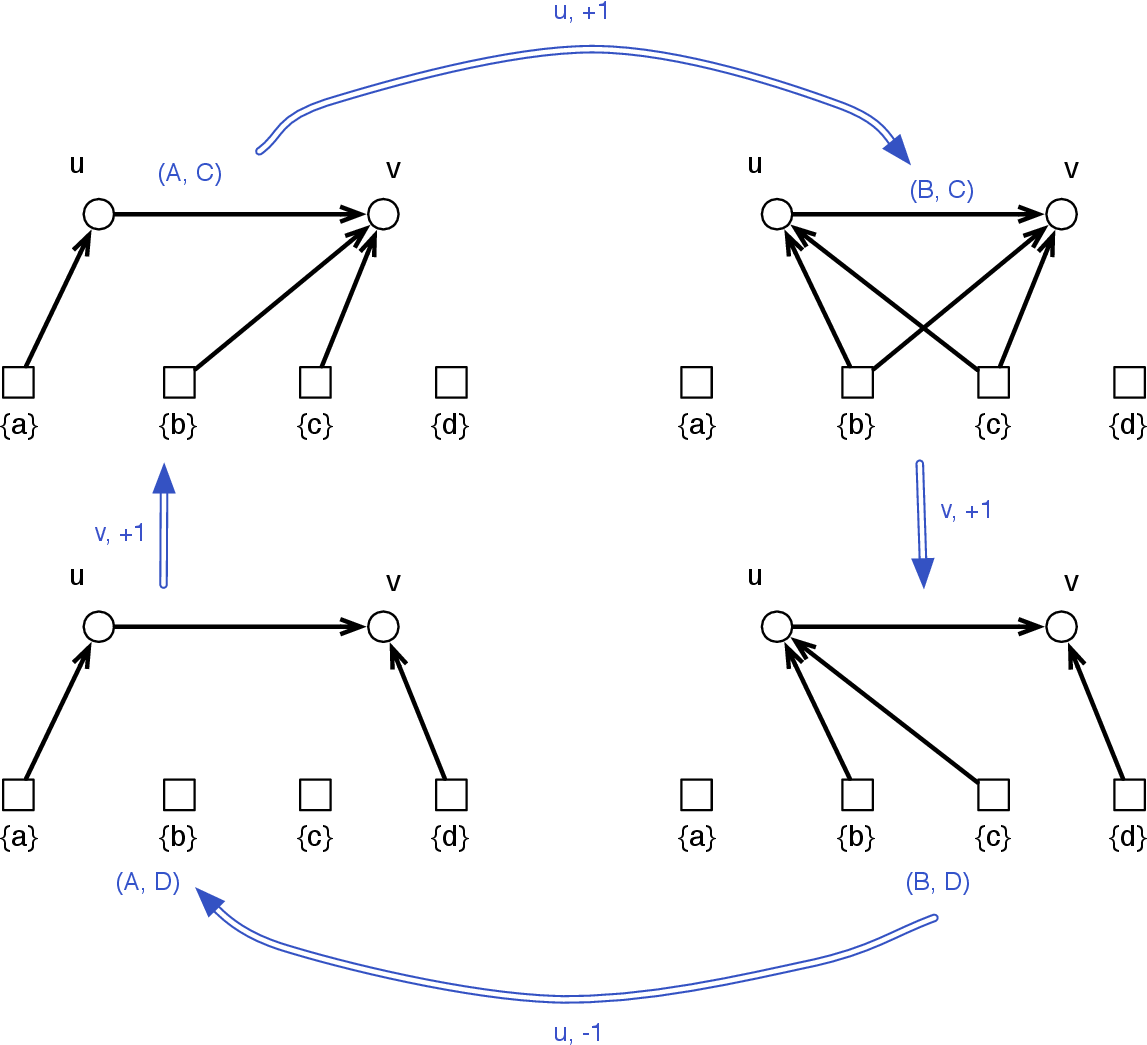}
\caption{A 4-cycle $(A,C)\rightarrow (B,C)\rightarrow (B,D)\rightarrow
  (A,D)\rightarrow (A,C)$ in the strategy space with non-zero utility
  variation. The cycle is represented with blue double arrows. Each double arrow
  corresponds to a move where a user changes his connection strategy.  This
  brings the network from one configuration to another. Each double arrow is
  labeled with the name of the user making the move and his utility variation.}
\label{fig:cycle4}    
\end{figure}

Instead of trying to obtain tight bounds on the convergence time of 
arbitrary sequences of moves, we will now
prove that convergence time is polynomial under some natural
assumption concerning the dynamics. 
The idea is to assume some fairness among users
in the sense that they regularly have the opportunity to make a move.
To measure this, we call \emph{round} a sequence of moves where each user
can be associated to a point in the sequence where he
either performs a selfish move or cannot make any selfish move 
(if he is given the possibility to make a move at that point, 
he can either perform the move indicated in the sequence or
no move can increase his utility in the configuration obtained by
the moves up to that point).
Interestingly, we could consider a sequence of moves as fair when
starting at any moment $t$ in the sequence, the moves from $t$ to 
$t+O(n)$ constitute a round. If convergence is polynomial in number of rounds,
it is then polynomial in number of moves also.
More generally, we consider as fair 
any sequence of moves without any infinite round.
Such sequences will be called \emph{fair dynamics} and can always be 
decomposed as a sequence of finite rounds.

\begin{proposition}
Any homogeneous flow game with $p$ producers and $n$ users having
average budget of attention $\overline\Delta$
converges in $O(np + n^2\overline\Delta)$ rounds
under fair dynamics.
\end{proposition}

\begin{proof}
We claim that the number of users with minimum utility $U_m$ (those getting the
least number $U_m$ of subjects) decreases every three rounds as long as
equilibrium is not reached and the value of $U_m$ has not increased. 
The proposition clearly follows from that fact since
$3n$ rounds at most then suffice to increase the minimum utility by one at least
and the maximum utility a user can get is bounded by $\min(p,
n\overline\Delta)$.  To show this, assume that during one round, no user with
minimum utility can make a selfish move. A first trivial case occurs when all
users have the same utility $U_m$. This means that no user could make a move
which implies that equilibrium has been reached. Otherwise, consider what occurs
in the next round. Either no node can move and we reach equilibrium, or a node
with minimum utility $U_m$ can move, or a node $u$ receiving more than $U_m$
subjects can move and then receives at least $U_m + 2$ subjects.  In the latter
case, consider the first time the opportunity of moving is given to a node $v$
receiving $U_m$ subjects and following directly at least one producer $x$. (Such
a node must exist: as any user followed by a user having minimum utility has
also minimum utility, if all users with minimum utility are only connected
together, their utility is zero and they can obviously move.) Then rewiring
$(x,v)$ into $(u,v)$ is a move for $v$. This is due to the fact that $u$ still
receives $U_m + 2$ subjects at least until the turn of $v$ comes (any move by a
user $u'$ with utility $U_m+1$ or more increases his utility to $U'\ge U_m+2$
and if $u$ is affected, his utility cannot drop bellow $U'$).  As $u$ receives
at least two subjects not received by $v$, the utility of $v$ increases by one
at least with this move. As the move of $u$ occurs within the second round after
a round without any progress for nodes with minimum utility, the move of $v$
occurs within the third round at most.
\end{proof}

The interested reader can easily build sequences of moves with length
$\Omega(n^2)$. We will thus not try to improve beyond polynomial time
convergence in this section. However, we will see in the heterogeneous
case how convergence within a logarithmic number of rounds can arise
when interests of users have some geometrical structure
(see Section~\ref{sec:doubling}).

\section{Heterogeneous interests}
\label{sec:het}

We now consider the more realistic case where users have differing
sets of interests.  
To make the model even more general, we
assume here that users weight independently
topics.  Let $W_u(s)$ denote the weight (or \textem{value}) of topic $s$ to
user $u$.  The objective of a user is now to maximize the sum of the values 
of subjects he receives.
We will consider user-interest sets $S_u \subseteq P$ that
include topics of non-zero value, that is $S_u = \{s:
W_u(s)>0\}$.  
Such
user-specific weights for topics represent a natural expertise or
focussed interest users may have on a subset of topics.
(Note that the model presented previously corresponds to 
$W_u(s)=1$ for $s\in S_u$, $W_u(s)=0$ for $s\notin S_u$.) 

\subsection{Price of Anarchy}

We now show that the price of anarchy of such a system may be
unbounded.  


\begin{proposition}
In a heterogeneous flow game with $n$ users having 
budget of attention $\Delta$ each,
the price of anarchy can be $\Omega\paren{\frac{n}{\Delta}}$.
\end{proposition}

\begin{proof}
  We show the result through an example, illustrated in
  Figure~\ref{fig:het.poa}.  For positive integer $k$, 
  consider a system with $n=2k$ users having 
  budget of attention $\Delta\ge 2$ each, and $p=2(\Delta-1)k$ producers.  
  We distinguish two set of users $\set{a_1,\ldots,a_k}$ and $\set{b_1,\ldots,b_k}$.
  Similarly, the producers are partitionned into groups $\set{A_1,\ldots,A_k}$
  and $\set{B_1,\ldots,B_k}$ where each $A_i$ (resp. $B_i$) contains $\Delta-1$ 
  producers. 

  As illustrated in Figure~\ref{fig:het.poa.int},
  each user $a_i$ has a value of $1$ for each topic in $A_i\cup B_i$ and additionally 
  the first element of each $A_j$ for $j\not= i$.  In the figure, this is represented by the solid red line.  Similarly, each
  user $b_i$ has a value of $1$ for each topic in $A_i\cup B_i$ and additionally 
  the first element of each $B_j$ for $j\not= i$. This is represented by the dashed blue line in the figure.  Users have a value
  of zero for all other topics. 
 
  Figure~\ref{fig:het.poa.opt} shows a benchmark configuration, with solid red edges for nodes of type $a$ and dashed blue edges for nodes of type $b$. In this configuration, users $a_i$, $i=1, \ldots ,k$ construct an oriented ring, and similarly users $b_i$, $i=1,\ldots ,k$ construct a separate oriented ring.  They use their remaining links to connect to producers.
User $a_i$ is then connected to $a_{i-1}$
  (with $a_0$ corresponding to $a_k$) and also to all producers in $A_i$ using the remaining $\Delta-1$ links.
  Similarly, user  $b_i$ is connected to $b_{i-1}$
  (with $b_0$ corresponding to $b_k$) and also to all producers in $B_i$. The corresponding utility is $n(n/2 + \Delta-2)$, 
   so that
  the optimal global welfare $\U^*$ satisfies $\U^* \ge n^2/2$.  
  
  Figure~\ref{fig:het.poa.ne} shows an equilibrium configuration, where each user
  $a_i$ (resp. $b_i$) connects to producers in $A_i$ (resp. $B_i$) using $\Delta-1$ links and to
  $b_i$ (resp. $a_i$) using one link. Note that neither user can gain by making a unilateral move since each of the other users (of indices $j \neq i$) can only provide one additional subject as opposed to the 
 $\Delta-1$ subjects they now receive from each other. 
The global utility here is $\U = n(2\Delta-2)\le 2n\Delta$, 
  and the price of anarchy is thus at least $\frac{n}{4\Delta}$.  
\end{proof}

\begin{figure}[htbp]
\centering
  \subfigure[Interest sets of $a_i$ and $b_i$: solid red line for $a_i$ and dashed blue line for $b_i$. ]{
  \label{fig:het.poa.int}
  \includegraphics[width=.7\linewidth]{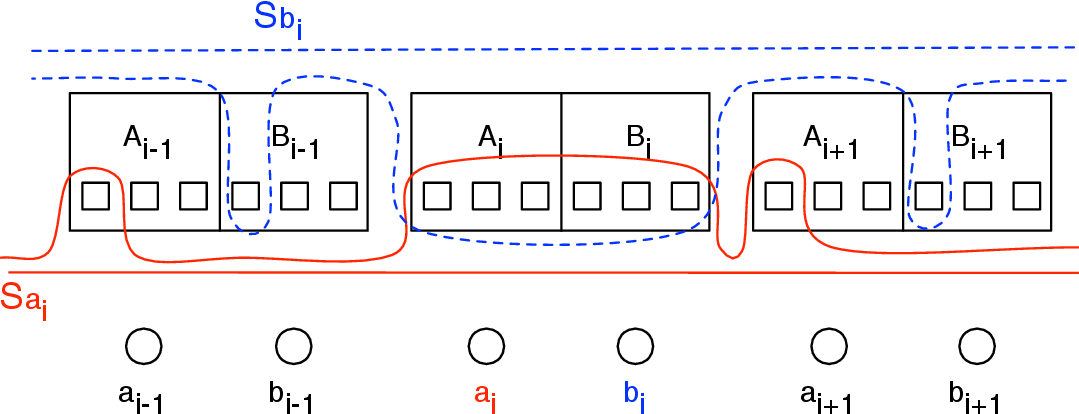}
  }\\%
  \subfigure[Benchmark configuration with two oriented rings.]{
  \label{fig:het.poa.opt}
  \includegraphics[width=.7\linewidth]{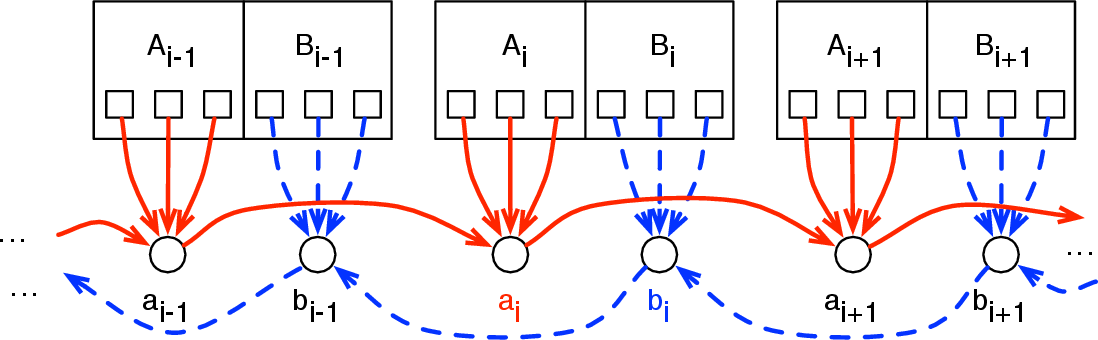}
  }\\%
  \subfigure[A pure Nash equilibrium configuration.]{
  \label{fig:het.poa.ne}
  \includegraphics[width=.7\linewidth]{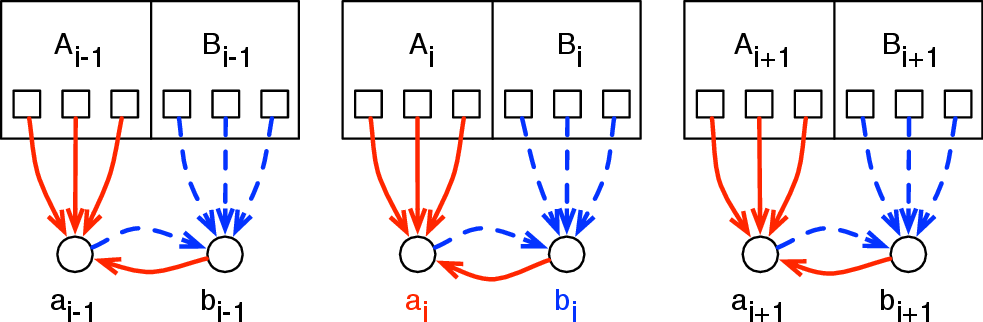}
  }
\caption{A heterogeneous flow game where all users have same budget of attention
$\Delta=4$ and two configurations showing that price of anarchy can reach
$\frac{n}{4\Delta}$.}
\label{fig:het.poa}
\end{figure}

\subsection{Convergence of dynamics}

We have shown that the price of anarchy can be unbounded with respect to
the number of users in some cases.  

We now show that selfish dynamics do not even guarantee convergence to an
equilibrium.  

\begin{proposition}
\label{prop:het.instability}
Selfish dynamics of a flow game with heterogeneous utilities may not converge.
\end{proposition}

\begin{proof}
Consider the following scenario with six retransmitting users $p_i, q_i, r_i$,
$i=1,2$, and two users $u_1, u_2$ each with degree $\Delta_i=3$.  The
retransmitting users publish sets of
topics as follows:  $p_1:\{a,b\}$, $p_2:\{c,d\}$, $q_1:\{x,y\}$,
$r_1:\{k,l\}$, $q_2:\{x,k\}$, $r_2:\{y,l\}$ .  The user-specific
values are
given in Table~\ref{tab:het}, where $\epsilon < 1/2$.   
As depicted in Figure~\ref{fig:het_instab}, each agent $u_i$ uses one connection
to follow user $p_i$ through whom he receives a total value of
$4$.  He also connect to the other user $u_{3-i}$ to receive another
topic of
value $2$ from $p_{3-i}$.  Now each user $u_i$ must select between
$q_1,q_2,r_1$ and $r_2$ for his third connection.  
We start with users $u_1$ and $u_2$
choosing $q_1$ and $q_2$ respectively. They thus receive
$8+\epsilon$  and $7+2\epsilon$  in total respectively.   User $u_2$ then
selects $r_2$ for receiving $l$ instead of $k$ and getting $8+\epsilon$.
This changes user $u_1$'s utility to $7+2\epsilon$.  
Then user $u_1$ can increase his utility by $1-\epsilon$, and does so by 
switching to $r_1$ for receiving $k$ instead of $x$. Now this decreases $u_2$'s
utility by $1-\epsilon$. This can indeed continue again and again as follows.
Denote the state of the system
by $(\mathcal{S}(u_1),\mathcal{S}(u_2))$ where $\mathcal{S}(u_i)$ is
user $u_i$'s strategy in selecting between $q_i$ and $r_i$.  
Under selfish moves, the
system may cycle as follows: $(q_1,q_2)$ $\rightarrow$ $(q_1,r_2)$ $\rightarrow$ $(r_1,r_2)$
$\rightarrow$ $(r_1,q_2)$ $\rightarrow$ $(q_1,q_2)$ $\rightarrow$
$(q_1,r_2)$ $\rightarrow \cdots$.
\end{proof}
\vspace{-0.5em}
\begin{figure}[htbp]
\centering
\includegraphics[width=\linewidth]{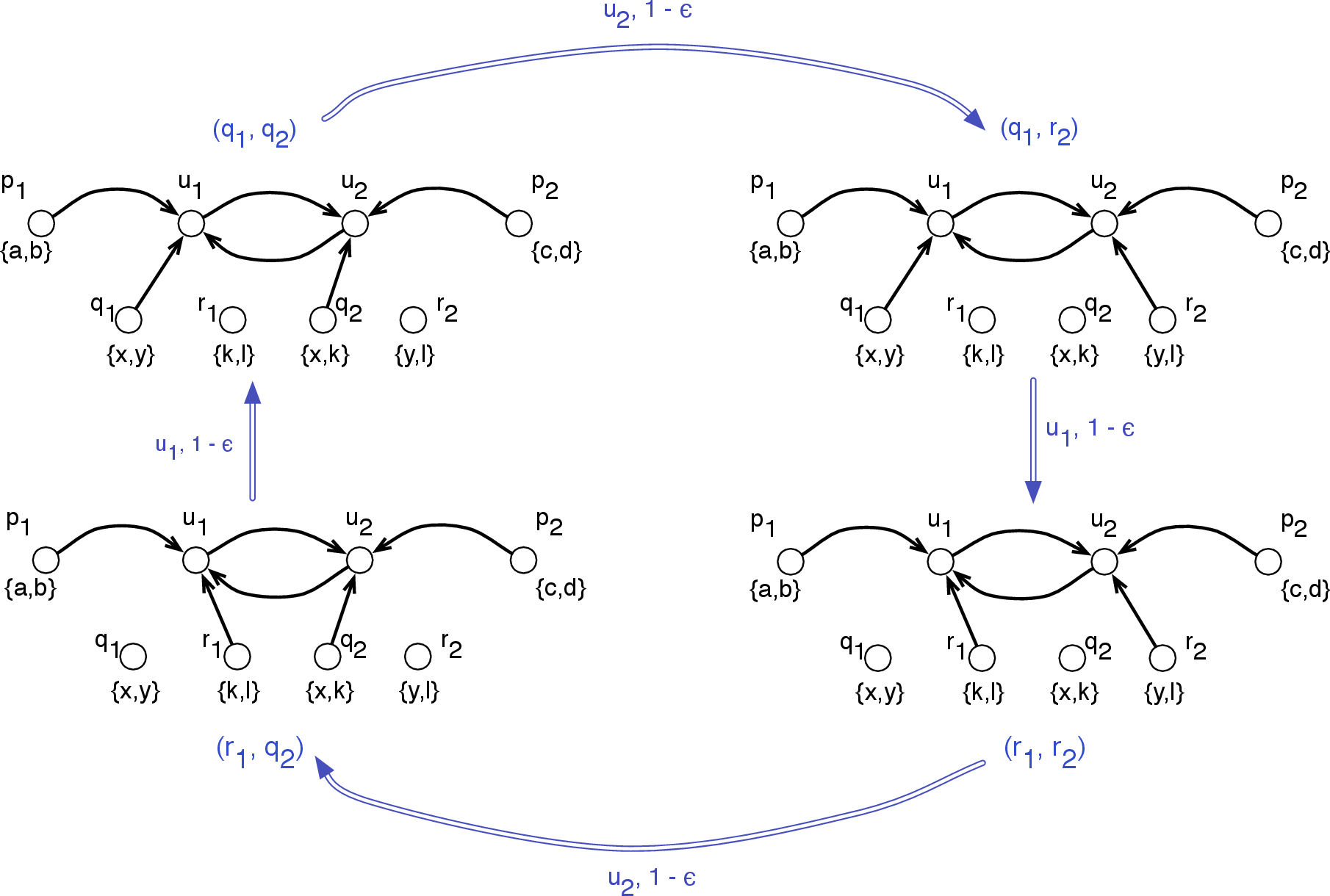}
\caption{Illustration of the proof of Proposition~\ref{prop:het.instability}.
  Instability with heterogeneous interest sets can arise with this
  4-cycle of selfish moves $(q_1,q_2)$ $\rightarrow$ $(q_1,r_2)$ $\rightarrow$
  $(r_1,r_2)$ $\rightarrow$ $(r_1,q_2)$ where users $u_1$ and $u_2$ change
  only one connection. The strategy of other users remains fixed. The set
  of subjects they receive is indicated.
  User-specific values for topics
  are those listed in Table~\ref{tab:het}. Each double arrow corresponds to a
  selfish move bringing from one configuration to another. It is labeled
  with the name of the user making the move and his utility variation.}
\label{fig:het_instab}
\end{figure}
\vspace{-0.5em}
\begin{table}[htbp]
\centering
\scalebox{1}{\bgroup\small
 \begin{tabular}{|l|c|c|c|c|c|c|c|c|}
\hline
   User\textbackslash Topic& \hspace{0.5em}$a$\hspace{0.5em} &
   \hspace{0.5em}$b$\hspace{0.5em} & \hspace{0.5em}$c$\hspace{0.5em} & \hspace{0.5em}$d$\hspace{0.5em} & \hspace{0.5em}$x$\hspace{0.5em} & \hspace{0.5em}$y$\hspace{0.5em} & \hspace{0.5em}$k$\hspace{0.5em} & \hspace{0.5em}$l$\hspace{0.5em} \\ \hline
$u_1$ & 2 & 2 & 2 & 0 & $\epsilon$ & 1 & 1 & $\epsilon$ \\ \hline
$u_2$ & 2 & 0 & 2 & 2 & 1 & $\epsilon$ & $\epsilon$ & 1 \\ \hline
 \end{tabular}
\egroup}
\caption{User-specific values for topics used in the proof of 
Proposition~\ref{prop:het.instability}.}
\label{tab:het}
\end{table}

With an arbitrary structure of user interest sets, we have thus shown
that the price of anarchy may be unbounded, and  dynamics may
not converge. The question of determining if
pure Nash equilibria exist is left open.

\section{Structured interest sets}
\label{sec:doubling}

We now revisit the efficiency of social filtering in an heterogeneous scenario, where interest sets are no longer arbitrary but instead are organized according to a well behaved geometry. Specifically we assume the following model.
A metric $d$ is given on a set $P'\supseteq P$ of subjects.  The interest set
$S_u$ of each user $u$ then coincides with a {\em ball} $B(s_u,R_u)$
in this metric, specified by a \emph{central subject} $s_u$ and a
\emph{radius of interest} $R_u$.  
We assume that the value of a subject
to a user is non-increasing in its distance from $s_u$.
Specifically, we assume $W_u(s)=f(d(s_u,s))$ for $d(s_u,s) \le R_u$, where $f(\cdot)$ is a
non-increasing positive function, and $W_u(s)=0$ otherwise.   
Without loss of generality, we
can assume $P'=\set{s_u : u\in V}\cup P$ and $S_u=B(s_u,R_u)\cap P$.
We shall first give conditions on the
metric $d$ and the sets $S_u$ under which an efficient configuration
exists. We will then introduce modified dynamics and filtering rules
which guarantee stability, i.e. convergence to an equilibrium.
A flow game where interest sets can be defined in this way
is called a \emph{metric flow game}.

The model can easily be generalized to more eclectic user interests
where topics a user is interested in correspond to the union of a
constant number of balls. We leave out the details of such
generalizations so as to keep the focus of the paper.  However, we
include a brief discussion later in the section, in the context of
Proposition~\ref{prop:doubling}.

\subsection{Sufficient conditions for optimal utility}

Consider the following properties of the interest set geometry.
\begin{enumerate}
\item\label{doubling} $\g$-doubling: $d$ is $\g$-doubling, i.e.
  for any subject $s$ and radius $R$, the ball $B(s,R)$ can be
  covered by $\g$ balls of radius $R/2$: there exists
  $I\subset S$ such that $\card{I}\le \g$ and $B(s,R)\subset
  \cup_{t\in I} B(t,R/2)$.
\item\label{covering} $r$-covering: $r$ is a covering radius, i.e. for each subject
$s \in P$ there is a user $u$ such that $d(s_u,s) \le r$ and $R_u \ge r$.
\item\label{sparsity} $(r,\d)$-sparsity: there are at most $\d$ subjects within
  distance $r$: $\card{B(s,r)}\le \d, \forall s$.
\item\label{smoothness} $r$-interest-radius smoothness: for any users $u,v$ with 
  $d(s_u,s_v) < 3R_u/2+r$, we have $R_v\ge R_u/2+r$ and $R_u\ge R_v/2+r$
  (users with similar interests have comparable interest radii). 
\end{enumerate}

Property~(\ref{doubling}) is a classical generalization of dimension
from Euclidean geometry to abstract metric spaces (an Euclidean space
with dimension $k$ is $2^{\Theta(k)}$-doubling). This is a natural assumption
if user interests can be modeled by proximity in a hidden low-dimensional space.
Property~(\ref{covering}) states that all subjects are within distance
$r$ from some user's center of interest and can thus be seen as
an assumption of minimum density of users' interests over the whole
set $P$ of available subjects.
Property~(\ref{sparsity}) puts an upper bound on the density of
subjects. In other words, we assume a level of granularity under which
we do not distinguish subjects.
Property~(\ref{smoothness}) is another form of smoothness assumption,
requiring that the radii of interest of nearby users do not differ too much.
This property is obviously satisfied if we assume that all users have
same radius of interest. In general,
it may seems debatable if we think of an expert 
next to an amateur. However, if we assume that a topic is split into several
subjects according to the level of expertise required to understand
the corresponding news, the assumption becomes
more natural as an expert is still interested in related subjects (with lower
level of understanding) and an amateur still has some focus if the correct
number of levels is considered.

We now show that an optimal solution exists, i.e. one in which each user
receives all subjects in his interest set, as each user $u$ has 
budget of attention at least 
$\g\d+\g^2\log\frac{R_u}{r}$. 

\begin{proposition}
\label{prop:doubling}
Consider a metric flow game satisfying
the $\g$-doubling, $r$-covering, $(r,\d)$-spar\-sity
and $r$-interest-radius smoothness
assumptions. If in addition each user $u$ has a budget of attention
at least $\g\d+\g^2\log\frac{R_u}{r}$, then there exists a collection
of user strategies allowing each user
$u$ to receive all subjects in $S_u$.
\end{proposition}

This result can easily be extended to the case where each user interest set is 
given by a union of balls (the number of balls being at most a 
constant $b$). It suffices to repeat the construction of the proof
for each ball, resulting in a factor $b$ in the resulting required 
budget of attention.
The assumptions have to be slightly modified
so that any subject is covered by some ball of a user (in the covering
assumption) and that two nearby balls have comparable radii (in the
smoothness assumption).  

\medskip
\begin{proof}
We define the ball $B_{u,i}:=B(s_u,\min(R_u,2^ir))$ for each user $u$ and each integer $i \ge 0$. 
The construction to follow will ensure that $u$ collects all subjects in $B_{u,i}$
through  a set $N_{u,i}$ of contacts such that 
$B_{u,i}\subset \cup_{v\in N_{u,i}} B_{v,i-1}$.

We first define $N_{u,1}=\set{p_s: s\in B_{u,1}}$.
According to the $\g$-doubling assumption, $B_{u,1}$ can be covered
by $\g$ balls of radius $r$. As the $(r,\d)$-sparsity implies that each
of these balls contains $\d$ subjects at most, the size of $N_{u,1}$
is upper bounded by $\g\d$.

Now, for $2\le i\le \ceil{\log\frac{R_u}{r}}$, the $\g$-doubling
assumption implies that $B_{u,i}$ can be covered by at most 
$\g^2$ balls of radius $2^{i-2}r$: there exists a set $L_{u,i}$ of at most
$\g^2$ subjects such that $B_{u,i}\subset \cup_{s\in L_{u,i}}
B(s,2^{i-2}r)$. From the $r$-covering assumption, we can then define
a set $N_{u,i}$ of at most $\g^2$ users such that each $s\in L_{u,i}$ is
at distance at most $r$ from some $s_v$ with $v\in N_{u,i}$.
We then have $B_{u,i}\subset \cup_{v\in N_{u,i}} B(s_v,2^{i-2}r+r)$.
Without loss of
generality, we can assume that for each $s\in L_{u,i}$, $B(s,2^{i-2}r)$
intersects $B_{u,i}$ (otherwise $s$ can safely be removed from 
$L_{u,i}$ as it does not cover anything useful). We thus have
$d(s_u,s)\le R_u+2^{i-2}r < 3R_u/2$ (note that $2^{i-1}r<R_u$
as $i\le \ceil{\log\frac{R_u}{r}}$). 
For $v\in N_{u,i}$ such that $d(s,s_v)\le r$, we then
have $d(s_u,s_v) < 3R_u/2+r$. From the $r$-interest-radius
smoothness, we then deduce $R_v\ge R_u/2+r>2^{i-2}r+r$, implying 
$\min(R_v,2^{i-1}r)\ge 2^{i-2}r+r$. The ball $B_{v,i-1}$ thus
contains $B(s_v,2^{i-2}r+r)\supset B(s,2^{i-2}r)$. 
Together with the definition of $L_{u,i}$,
this proves $B_{u,i}\subset \cup_{v\in N_{u,i}} B_{v,i-1}$.

The connection graph $G$ results from connecting each user $u$ to all contacts in the set 
$\cup_{1\le i\le\ceil{\log\frac{R_u}{r}}} N_{u,i}$.

\paragraph{Flow correctness:}
We show by induction on $i$ that each user $u$ receives all subjects in $B_{u,i}$. The direct connection to producers for
subjects in $B_{u,1}$ ensures this for $i=1$. For $i>1$, the
induction hypothesis implies that each user $v\in N_{u,i}$ receives
all subjects in $B_{v,i-1}$. From 
$B_{u,i}\subset \cup_{v\in N_{u,i}} B_{v,i-1}$, we conclude that $u$
will receive news about subjects in $B_{u,i}$ from its contacts in
$N_{u,i}$. As $S_u=B_{u,\ceil{\log\frac{R_u}{r}}}$, we finally know that $u$
receives all subjects in $S_u$.

\paragraph{In-degree bound:}
First, we have $\card{N_{u,1}}\le \g\d$. This comes from the fact
that $B_{u,1}$ is included in at most $\g$ balls of radius $r$ from
the $\g$-doubling assumption, and each of these balls contains at
most $\d$ subjects from the $(r,\d)$-sparsity assumption.
Second, we have already seen that $\card{N_{u,i}}\le\g^2$ for
$2\le i\le \ceil{\log\frac{R_u}{r}}$. We thus obtain the bound
$\g\d+\g^2\paren{\ceil{\log\frac{R_u}{r}}-1}<\g\d+\g^2\log\frac{R_u}{r}$.
\end{proof}

The core of the construction consists in covering a given ball radius of $2^ir$
with a set of $\g$ balls of radius $2^{i-1}r$.
As a covering set of $\g^2$ balls
can be computed through a simple greedy covering
algorithm~\cite{DBLP:conf/compgeom/Har-PeledM05}, a solution where
the required budget of attention is within a factor $\g$ 
from the bound of Proposition~\ref{prop:doubling} can thus
be computed in polynomial time. 

Note that a logarithmic number of connections allows to gather a polynomial number of subjects.
As previously mentioned, a budget of attention of
$\Delta=\g\d+\g^2\log\frac{R_u}{r}$ for each user $u$ is enough for
maximum utility. On the other hand, the number of subjects in $B(s_u,R_u)$
can be polynomial in $R_u$. For example, if the subjects are placed regularly
in a $d$ dimensional lattice, it would be in the order of $R_u^d$
(the doubling assumption ensures that it is at most polynomial). 
A logarithmic number of connections is thus sufficient
for gathering the subjects interesting a user.
Thus this configuration gives substantial savings in
comparison to one where users would connect directly to all their
subjects.

Clearly the configuration graph identified in this theorem is 
an equilibrium: as maximum
utility is reached, no user can increase its utility by
reconnecting.  We now study conditions that guarantee convergence of 
dynamics.

\subsection{Sufficient conditions for stability}

We first define two rules regarding republication of subjects received and reconnections.
\begin{enumerate}
\item\label{expertise} Expertise-filtering rule:
  when a user $u$ is connected to a user $v$, $u$ only receives
  subjects $s$ such that $d(s_v, s) \le d(s_u, s)$. 
\item\label{priority} Nearest-subject rule for re-connection:
  when reconnecting, each user $u$ gives priority to subjects that are
  closer to $s_u$: a new subject $s$ is gained by $u$ so that no subject
  $t$ with $d(s_u,t) \le d(s_u,s)$ is lost. (On the other hand,
  any subject $t$ with $d(s_u,t) > d(s_u,s)$ can be lost.)
\end{enumerate}

Rule~\ref{expertise} can be interpreted as follows. The center of expertise of
a user is the same as its center of interest, and the distance $d$ also captures expertise of users about subjects, in that $u$ is more expert than $v$ on subject $s$ if and only if $d(s_u,s)\le d(s_v,s)$. The rule then amounts to a sanity check where $u$
discards news from sources that have less expertise than himself on the
subject. 
We capture this with the following slight variation of the model.
A flow game \emph{with expertise-filtering} is a flow game where
reception of a subject $s$ by user $u$ occurs only when there exists
a directed path $s=u_0,\ldots,u_k=u$ from $s$ to $u$ such
that for each $1\le i< k$, $s\in S_{u_i}$ (i.e. $d(s_{u_i},s)\le R_{u_i}$) 
and $d(s_{u_i},s)\le d(s_{u_{i+1}},s)$.

Rule~\ref{priority} states that a user $u$ prefers to receive a
subject he is more interested in (i.e. closer to $s_u$) rather than
any number of subjects that are less interesting.  A flow game is denoted to
be \emph{with nearest-subject priority} if the utility function of
each user $u$ is defined by
$U_u(F)=\max\set{R : u \text{ receives all } s\in B(s_u,R)}$.

\begin{proposition}
\label{prop:stable}
Any metric flow game with expertise-filtering and nearest-subject priority
has an ordinal potential function, implying that selfish dynamics always
converge to an equilibrium in finite time.
\end{proposition}

The proof shows the existence of an ordinal potential function.  

\begin{proof}  
 Consider the set $\D=\set{d(s_u,s) : u\in V, s\in {P}}$ of all possible
 distances from the central subject of any user to any subject.  Let
 $r_1,\ldots,r_m$ denote all elements of $\D$ sorted in non-decreasing order
 (i.e. $r_1\le\cdots\le r_m$) with ties broken arbitrarily.  Let $n_i$ denote
 the number of pairs $(u,s)$ such that $d(s_u,s)=r_i$ and $u$ receives
 $s$. Consider the tuple $(n_1,\ldots,n_m)$. When a user $u$ makes a selfish
 move, he increases his utility by receiving a new subject $s$.  Let $i$ denote
 the index such that $d(s_u,s)=r_i$. Any lost subject $t$ must satisfy $d(s_u,t)
 > d(s_u,s)$ by the nearest-subject rule. If a lost subject $t$ was received by
 some user $v$ through a path from $u$ to $v$, we have $d(s_v, t) \ge d(s_u, t)$
 by the expertise-filtering rule.  We thus deduce $d(s_v,t) > d(s_u, s)$,
 implying that $n_j$ can decrease only for $j>i$.  The tuple $(n_1,\ldots,n_m)$
 thus increases according to the lexicographical order after any selfish move.
 As the size of $\D$ is at most $np$, the product $np$ is also a trivial
 upper bound on each $n_i$.  The tuple
 $(n_1,\ldots,n_m)$ can thus be read as a number in base $np$.  This number
 always increases under selfish moves, implying that 
 $\sum_{0\le i\le m}n_i\,(np)^{m-i}$ is a potential function.
 This potential function always increases until a local maximum
 is reached, proving convergence to an equilibrium.
\end{proof}


We can additionally
prove fast convergence under sufficient conditions for optimal utility
and fair dynamics under best response. 
We call \emph{best response} a move
where a user $u$ gets the best possible utility (with nearest-subject priority) 
given the current connections of other users. In other words, $u$
receives all subjects within distance $R$ from $s_u$ after the move
and no move could provide all subjects within distance $R'$
from $s_u$ with $R'>R$. 
Recall that a sequence is fair if it can be decomposed in a sequence of
rounds where each user has the opportunity to make a move during each round
as defined in Section~\ref{sub:convhom}. If in addition users only perform
best response moves, we say that the system
is under \emph{fair dynamics with best response}.

We are now ready to prove the following:
\begin{theorem}
\label{th:stab.opt}
Consider a metric flow game with expertise-filtering and nearest-subject priority
that satisfies the $\g$-doubling, $r$-covering, $(r,\d)$-sparsity
and $r$-interest-radius smoothness assumptions. If in addition each
user $u$ has budget of attention at least $\g\d+\g^2\log\frac{R_u}{r}$,
selfish dynamics converge to an equilibrium where
each user $u$ receives all subjects in $S_u$, implying that the
price of anarchy is then 1. Moreover, the system converges in at most
$\ceil{\log\frac{R_m}{r}}$ rounds under fair dynamics with best response
where $R_m$ is the maximum radius of interest over all users.
\end{theorem}

\begin{proof}
Consider a configuration where some user $u$ does not receive some
subject $s$ in his interest ball. Such a pair $(u,s)$ is called an 
unsatisfied pair. Without loss of generality we
consider an unsatisfied pair $(u,s)$ with
smallest $d(s_u,s)$ among all unsastisfied pairs. 
Let $j$ be the smallest integer such that $d(s_u,s)\le 2^j r$ holds.
As in the construction of the proof of Proposition~\ref{prop:doubling}, 
user $u$ can then receive all subjects in $B_{u,i}=B(s_u,\min(R_u,2^ir))$
through connections to the nodes in some set $N_{u,i}$ as follows.
The set $N_{u,1}$ contains at most $\g\d$ producers: those within
distance $2r$ from $s_u$. For
$2\le i\le j$, the set $N_{u,i}$ contains at most $\g^2$ users
such that $B_{u,i}$ is included in $\cup_{v\in N_{u,i}}B_{v,i-1}$.
Following these users allows $u$ to receive all subjects in 
the ring $B_{u,i}\setminus B_{u,i-1}$. The reason is twofold. 
First, the  choice of $(u,s)$ ensures that every user $v$ receives 
all subjects in $B_{v,i-1}$ as this ball has radius 
$2^{i-1}r$ at most and $2^{i-1}r \le 2^{j-1}r < d(s_u,s)$ by the choice
of $j$.  Second, any subject $s\in B_{v,i-1}\setminus B_{u,i-1}$ 
where $v$ is a user in $N_{u,i}$ is received by $u$ according to 
expertise-filtering since $d(s_v,s)\le 2^{i-1}r$ and $d(s_u,s)>2^{i-1}r$.
Overall, $u$ can receive all subjects within distance $\min(R_u, 2^j r)$
including $s$.
Nearest-subject priority implies that the configuration is unstable as
long as $\Delta_u \ge \g\d+\g^2 (j-1)$ which is the case for 
$\Delta_u\ge \g\d+\g^2 \log\frac{R_u}{r}$ since $R_u\ge d(s_u,s) > 2^{j-1} r$.
Since the system must stabilize to some equilibrium according to
Proposition~\ref{prop:stable}, every user $u$ must receive all news about
subjects in $S_u$ in that stable configuration.

\paragraph{Convergence speed:}
Let $P_j$ denote the property that every user $u$ receives all subjects in
his ball of radius $\min(R_u,2^j r)$.
We show by induction on $j$ that $P_j$ is satisfied 
after the first sequence of best response moves constituting $j$ rounds.
Consider the first moves. The nearest-subject priority rule
ensures that each user $u$ receives all subjects 
in $B_{u,1}$ after his first move.
The reason is simply that his budget is sufficient to connect directly
to all producers in $B_{u,1}$ (recall that this ball has size at most
$\g\d$ as shown in the proof of Proposition~\ref{prop:doubling}).
A move later on by a user $v$ cannot alterate the reception of a subject
$s$ with $d(s_u,s)\le 2r$. This is an effect of the expertise filtering
rule: $u$ can be affected only when he receives $s$ by a path from $v$ to $u$
with $d(s_v,s)\le d(s_u,s)\le 2r$ according to expertise filtering
and we know that a best response move of $v$ ensures that $v$ will receive 
all subjects in $B(s_v,2r)$ after the move.
Property $P_1$ is thus satisfied as soon as every user has made a move
under best response, that is after the first round.
Now assume that $P_{j-1}$ is satisfied.
A move by user $u$ cannot incur the loss of a subject $s$ for
a user $w$ whose central subject $s_w$ is at distance at most $2^{j}r$ from $s$.
The reason is that if $w$ receives this subject through a path from $u$,
the expertise-filtering rule implies that $s_u$ is at distance at most
$2^{j}r$ from $s$ also. On the other hand, $P_{j-1}$ implies the necessary
conditions to apply the same argument as in the first part of the proof.
We can thus show that some move by user $u$ will allow him to receive
all subjects within distance $\min(R_u,2^j r)$.
As $u$ forwards this subject before the move, we have 
$d(s_u,s)\le \min(R_u,2^j r)$ and $u$ still forwards the subject after a best
response move. We can thus conclude that if a user $w$ receives all subject
within distance $2^{j}r$, he will continue to receive all of them along
the $j$th round. This implies in particular that
property $P_{j-1}$ thus remains satisfied along the round.
Additionally, a user $u$ receives all subjects in his ball of radius 
$\min(R_u,2^j r)$ after his first move in the $j$th round and this
is preserved during the sequel of the round.
We can thus conclude that $P_j$ is satisfied as soon as
the $j$th round is completed and remains satisfied afterwards.
This completes the proof by induction.
For $j=\ceil{\log\frac{R_m}{r}}$,
$P_j$ imply that every user receives all subjects in his interest ball.
The convergence time is thus at most $\ceil{\log\frac{R_m}{r}}$ rounds.
\end{proof}

\section{Budget of attention and cost of connections}
\label{sec:costs}

As a simplifying assumption, we have considered up to now
that filtering each connection had the same cost. We now discuss
how our work can be extended to reflect the fact that the connection $(v,u)$
from a user $u$ to a user $v$ depend on how the interests of $u$ and $v$
differ.
A simple idea would be to let the cost be an increasing function of the
number of uninteresting messages $v$ sends to $u$. However, this would
make the model much more complex as costs would then depend on the dynamics.
Moreover this would not reflect the reality where a link is usually 
established on a long term basis. We thus propose to better model the cost
of attention of a connection as an increasing function of the number
of uninteresting subjects $v$ may potentially bring to $u$, that is
$|S_v\setminus S_u|$. If we normalize the cost of connecting directly
to a producer to 1, a simple cost function for establishing link $(v,u)$
could be $c(v,u)=1+\alpha_u |S_v\setminus S_u|$ where $\alpha_u>0$ is some 
parameter comparing the cost of filtering an uninteresting subject for user $u$
to the cost of initiating a connection.

The model remains the same in the homogeneous case. In the heterogeneous
case, the negative results of Section~\ref{sec:het} remain valid in this
more complex model as we could expect. The example with high price of
anarchy given in Figure~\ref{fig:het.poa.int} can be modified so that
the cost of connecting to any user is the same and the bad equilibrium
configuration remains stable (it suffices to add $(n/2-1)-(2\Delta-3)$ 
subjects for each pair $a_i,b_i$ of users that interest both of them and 
no one else). In the non-convergence example of Figure~\ref{fig:het_instab}, 
the two users that oscillate
between two strategies are basically interested in the same subjects and
they oscillate between users bringing only interesting subjects.
The possibility of non-convergence thus remains valid also.

Pushing forward the idea, we can assume that a user tends to accept a
certain fraction of uninteresting content compared to interesting content.
This could be modelled by setting $\alpha_u = \frac{\beta}{|S_u|}$ for some
constant $\beta > 0$. Additionally, there is no reason for counting several
times a subject that is brought by several connections (micro-blogging
systems can automatically eliminate duplicates). 
We can thus estimate globally the cost of the set of connections $F_u$
made by user $u$ as:
$$
c(F_u) = |F_u| + \beta \frac{\card{\cup_{(v,u)\in F_u}S_v \setminus S_u}}{\card{S_u}}
$$

Our model with structured interest sets naturally fits with this kind of
cost if we make a slightly stricter assumption on the metric, namely that it has
bounded growth. Given a constant $\gamma'>1$,
a metric is $\gamma'$-growth-bounded if for any point $s$ and
radius $R$, the ball $B(s,2R)$ is larger than $B(s,R)$ by a
multiplicative factor of $\gamma'$ at most. This is indeed a special case
of doubling metric and still generalizes Euclidean metrics.
The expertise-filtering rule implies that a user $u$ follows users at distance
at most $2R_u$. We can adapt the interest-radius smoothness 
assumption by requiring that for any users $u,v$ with
$d(s_u,s_v)\le 4 \max(R_u,R_v)$, we have $R_v\ge R_u/2$ and $R_u\ge R_v/2$. 
This setting thus implies $\cup_{(v,u)\in F_u}S_v\subseteq B(u,4R_u)$,
and the size of this ball is at most $\gamma'^2$ times larger than 
$B(u,R_u)=|S_u|$ by the $\gamma'$-growth-bounded hypothesis.
We thus get $c(F_u)\le |F_u| + \beta {\gamma'}^2$.
Expertise-filtering and smoothness assumptions on the metric modeling the
interests thus imply that the cost term for the filtering of uninteresting
content is upper-bounded by a constant. 
The results presented in Section~\ref{sec:doubling} thus still
apply up to the corresponding additive term in the budget of attention bounds.
We thus see that this finer model gives
another justification to expertise-filtering. This may indeed reveal that
the cost of filtering may naturally induce an incentive for applying
expertise-filtering.

\section{Concluding remarks}
\label{sec:conc}


We have shown that a flow game can have complex dynamics that may not
converge. 
However, we can prove convergence to efficient equilibrium for both
homogeneous flow games (with very weak assumptions) and metric
flow games (with more technical assumptions).
While our proofs give exponential bounds on convergence time in general,
we get linear convergence time up to a logarithmic factor
(in number of moves) for structured
interest set with expertise-filtering and nearest-subject priority,
showing that understanding the structure of interests and its relation
to forwarding mechanisms is a key aspect of information flow in social networks.
Direct follow up of this work concerns the study of the speed of
convergence in general and the characterization of flow games having 
pure Nash equilibria.

%

A dual variant of our model could be to consider that every user
gathers all the subjects he is interested in while he tries
to minimize the required cost of attention. We could also mix both
models, using utility functions combining coverage of interest set
and cost of attention (the function being increasing in the
number of interesting subjects received and decreasing in the costs
of attention of the formed links).
Another interesting variant resides in considering the size of flows
or equivalently their rate of news. The budget of attention required
to follow a flow should then increase accordingly to its size. 
This variant is complementary to weighting flows as a flow with more news
might be weighted higher by users wishing to follow it.


In that context, we believe the two following directions
are promising for efficient social dissemination. First, incentive
mechanisms, e.g. reputation counters maintained by users, or payments
between users, may be a complementary approach to augment the performance of
self-organizing social flows. Second, more elaborate content filtering
between contact-follower pairs may also lead to substantial
improvements. We have already introduced expertise filtering, which
could translate into implementable mechanisms in existing social
networking platforms. More generally there appears to be a rich design
space of filtering rules based on combinations of interests and
expertise.

\bibliographystyle{abbrv}
\bibliography{flow_social}

\end{document}